\RequirePackage{fix-cm}
\documentclass[smallextended]{svjour3}       % onecolumn (second format)
\smartqed  % flush right qed marks, e.g. at end of proof
\usepackage{graphicx}
\usepackage{amsmath,amssymb,amsthm}

\begin{document}

\title{Existence of a Convex Polyhedron with Respect to the Given Radii\thanks{This research was supported by Chiang Mai University, Thailand.}
}
%\titlerunning{Short form of title}        % if too long for running head

\author{Supanut Chaidee         \and
        Kokichi Sugihara %etc.
}

%\authorrunning{Short form of author list} % if too long for running head

\institute{Research Center in Mathematics and Applied Mathematics, Department of Mathematics, Faculty of Science, Chiang Mai University \at
	239 Huaykaew Road, Suthep District, Muang, Chiang Mai 50200, Thailand \\
	Tel.: +66-53-943326\\
	\email{supanut.c@cmu.ac.th}           %  \\
	%             \emph{Present address:} of F. Author  %  if needed
	\and
	Meiji Institute for Advanced Study of Mathematical Sciences, Meiji University \at
	4-21-1 Nakano, Nakano-ku, Tokyo, Japan, 164-8525\\
	\email{kokichis@meiji.ac.jp} 
}

\date{Received: date / Accepted: date}
% The correct dates will be entered by the editor

\maketitle

\begin{abstract}
Given a set of radii measured from a fixed point, the existence of a convex configuration with respect to the set of distinct radii in the two-dimensional case is proved when radii are distinct or repeated at most four points. However, we proved that there always exists a convex configuration in the three-dimensional case. In the application, we can imply the existence of the non-empty spherical Laguerre Voronoi diagram.
% \PACS{PACS code1 \and PACS code2 \and more}
% \subclass{MSC code1 \and MSC code2 \and more}
\end{abstract}

\section{Introduction}
Suppose that we are given a set of $n$ points $\mathcal{V}=\{v_1, v_2, ..., v_n\}$ in two-dimensional and three-dimensional spaces. One of the fundamental questions in computational geometry is to consider the convexity of the given set, such as computing the convex hull of $\mathcal{P}$. When $\mathcal{V}$ is finite, the convex hull of a set $\mathcal{V}$ is a polygon in the two-dimensional case and polyhedron in the three-dimensional case. The problem of algorithmic construction of a convex hull was initially addressed by Preperata \cite{Preperata1977}.

It is well-known that the convex hull is the primitive object in the computational geometry. For example, the construction of spherical Voronoi diagram and spherical Laguerre Voronoi diagram, as defined in \cite{Sugihara2002}, uses the central projection of 3D convex hull onto the sphere to generate the Delaunay diagrams as presented in \cite{Sugihara2000}. 

In the case of the spherical Voronoi diagram, the points for the computed 3D convex hull are on the sphere. Therefore, the central projection of the 3D convex hull consists of all Delaunay triangulation of the diagram. However, the spherical Laguerre Delaunay diagram construction is different to the ordinary spherical Voronoi diagram in such a way that each generator contains its weight, and the points for generating the convex hull can be shifted over the sphere. Therefore, the convex hull of those points may include some points inside the constructed convex hull. Since the diagram can be constructed from the central projection of the convex hull onto the sphere, the Laguerre cell corresponding to the hidden point is empty, which is a dilemma of the spherical Laguerre Voronoi diagram.

Suppose that there is a set of weights of the spherical Laguerre Voronoi diagram $W=\{w_1, ..., w_n\}$. We would like to find the location of generator position $\mathcal{P}=\{p_1, ..., p_n\}$ on the unit sphere $S^2$ in such a way that the no cell of the generated spherical Laguerre Voronoi diagram is empty. This problem can be transformed to the following problem.

Let $\mathcal{R}=\{r_1, ..., r_n\}$ be a set of radii from the origin $O$. We would find the configuration of all points $\mathcal{V}=\{v_1, ..., v_n\}$ such that all of the points are vertices of a convex polyhedron. 

We firstly review the similar and related problems to our study.

\subsection{Related works}
To consider the literature, we primarily focus on the problems of convexification and convex configuration in the two-dimensional case.

Suppose that there is a closed chain composing of the vertices and links. The \textit{reconfiguration problem} is a problem to consider whether or not the given configuration can be transformed into another configuration. Lenhart and Whitsides \cite{Lenhart1995} considered the problem when the lengths of links are fixed, and the reconfiguration is allowed to across other links. This result also proved that every polygon could be convexified, i.e., the edge lengths of resulting convex polygon is preserved.

The more specified problem to the reconfiguration problem is the \textit{polygon convexification} problem, a problem to transform a configuration of the simple polygon in the initial stage to a convex polygon. Everett et al. \cite{Everett1998} considered the polygon convexification problem in the case of star-shaped polygon and proved that every star-shaped polygon in general position could be convexified. In this problem, the lengths of the links are not necessary to be fixed.

One of the famous problems called the carpenter's rule problem is to ask whether we can continuously move a simple polygon in such a way that all vertices are in convex position. Aichholzer et al. \cite{Aich2001}, Connelly et al. \cite{Connelly2003} studied the problem to convexify the polygonal cycle by employing a continuous motion to be a convex closed curve such that no links cross each other during the motion. Especially, the study in \cite{Aich2001} defined the term \textit{convex configuation} as the configuration of a convex polygon where edge links are fixed.

In the three-dimensional space, based on our observation, the configuration problem of points in 3D to be a convex set, has not clearly identified yet. However, in the general dimension, the \textit{convex hull frame problem}, known as \textit{redundancy removal problem}, is a problem to compute vertex description of the given set of points. That is, to justify whether a point is in a convex hull of the given set. If it is inside the convex hull, we remove that point. 

Clarkson \cite{Clarkson1994}, Ottman, et al. \cite{Ottmann1994} considered the algorithms for testing whether a given point is inside the convex hull or not. Dula and Helgason \cite{Dula1996} studied the problem by identifying the extreme points (or vertices) of the convex hull of the given points using the linear programming viewpoint. Other similar problems were the vertex enumeration of the convex hull as presented by in \cite{Kalantari2015}.

With the basic problem of the convex hull frame problem, the closest issues to the Voronoi diagram in Laguerre geometry were firstly addressed by Aurenhammer \cite{Aurenhammer1987} and Imai et al. \cite{Imai1985}. In \cite{Imai1985}, the emptiness of the Laguerre Voronoi cell in the Euclidean space $\mathbb{R}^d$ was identified that the Voronoi polygon of the generating circle $c_k$ is empty if the center of circle $c_k$ is not on the boundary of the convex hull.

In the spherical case, assume that all points were on or close to a sphere. Carili et al. in \cite{Caroli2010} established the sufficient condition under which no point is hidden in other planes of the convex hull with respect to other points.

\subsection{Problem statement and our contribution}
In this study, we investigate the modification of the previous convex configuration problem. Suppose that a set of radii is given with a fixed point. We would find the existence of a convex polyhedron whose vertices correspond to the given set of radii.

In two dimensional case, the convex configuration of points is a polygon whose the edge lengths of a polygon are allowed to be moved, and fixed for the radii, whereas the problems in \cite{Lenhart1995,Everett1998,Connelly2003} fixed the link lengths. 

The problem in the two-dimensional case is generalized to the three-dimensional case, i.e. we find a convex polyhedron satisfying the given radii set. The main motivation of this study is initiated from the non-empty property of the spherical Laguerre Voronoi diagram which the problem can be simplified to the problem of the modified convex configuration problem in the three-dimensional space. The existence of the convex configuration can guarantee that for any set of weights, we can always find a spherical Laguerre Voronoi diagram whose all Voronoi cells are nonempty, which is the different approach to the problems stated in \cite{Clarkson1994,Ottmann1994,Dula1996,Kalantari2015}. 

This paper is organized as follows. In Section 2, the notation, definitions, and the formulation of problems are provided. The existence of a convex polygon which is a convex configuration in the two-dimensional case is discussed in Section 3. In Section 4, the existence of the convex configuration in the three-dimensional is proved. The application of the problem to the spherical Laguerre Voronoi diagram, which answers the question from the motivation of the study, is described in Section 5. The concluding remarks and future study will be clarified in the last section.

\section{Preliminaries}

In this section, we define the notations and the necessary definitions. After that, we formulate the problem.
\subsection{Notations and Definitions}
Firstly, we mainly focus on the definitions in the two-dimensional case. The definitions in the three-dimensional case will be provided in the later part.

Let $V=\{v_1, ..., v_n\}$ be a set of vertices which is arranged counterclockwise on the plane. An edge $e_i=(v_{i}, v_{i+1})$ is a segment joining vertices $v_i$ and $v_{i+1}$ with the length $l_i:=d(v_{i}, v_{i+1})$, where $d(v_{i}, v_{i+1})$ denotes the distance between $v_i$ and $v_{i+1}$.

A \textit{chain} is a straight line graph formed by the set of edges $\mathcal{E}=\{e_1, ..., e_{n-1}\}$. A \textit{polygon} $P$ is a closed region bounded by a closed chain generated from the set of edges $\{e_1, ..., e_{n}\}$, where $e_i=(v_i, v_{i+1})$. A polygon $P$ is said to be \textit{simple} if the chain does not intersect itself except the vertices of $P$.

Let $e_i=(v_i, v_{i+1})$ and $e_{i+1}=(v_{i+1}, v_{i+2})$ be adjacent edges of a polygon $P$ whose common vertex is $v_{i+1}$. The angle between $e_i$ and $e_{i+1}$ is denoted by $\angle v_iv_{i+1}v_{i+2}$ and is measured clockwise  from the segment $\overline{v_iv_{i+1}}$.

A polygon $P$ is said to be \textit{convex} if and only if for any point $p, q$ in the polygon $P$, a segment joining $p$ and $q$ is in $P$. Also, for each angle $\angle v_iv_{i+1}v_{i+2}$ of $P$, where $i=1, ..., n, v_{i+1}=v_1, v_{i+2}=v_2$, $\angle v_iv_{i+1}v_{i+2} \leq \pi$ if and only if $P$ is convex. Remark that it is impossible that $\angle v_iv_{i+1}v_{i+2} = \pi$ for all $i$. For the special case, a polygon $P$ is said to be a \textit{strictly convex polygon} if and only if $\angle v_iv_{i+1}v_{i+2} < \pi$ for all $i$.

For a given edge length set $L=\{l_1, ..., l_n\}$, a \textit{convex configuration} of edge lengths is a convex polygon whose length of edges satisfy the set $L$ with counterclockwise order. 

The \textit{radius} $r_i$ of a polygon vertex $v_i$ is defined as the distance between the vertex $v_i$ and the given fixed point. Without loss of generality, we assume that the origin $O$ is such the fixed point.

For a given straight line $\ell$, an arbitrary half-plane with respect to the line $\ell$ is denoted by $H(\ell)$. The half-plane including the origin is written as $H^0(\ell)$.

Next, we generalize the mentioned definitions in the three-dimensional spaces. 

Assume that $\mathcal{V}=\{v_1, ..., v_n\}$ is a set of points in the three-dimensional spaces. In our context, the \textit{convex polyhedron} is a convex hull of a set $\mathcal{V}$. We can also construct a polyhedron from the intersection of a finite number of half-spaces. In this study, we focus on the polyhedron which is formed from the bounded intersection of half-spaces.

Similar to the two-dimensional case, without loss of generality, the radius $r_i$ of a polyhedron vertex $v_i$ is defined by the Euclidean distance between $v_i$ and the origin $O$.

In spherical geometry, we consider a unit sphere $S^2$ where the center is located at the origin. For $p, q\in S^2$, let $\tilde{d}(p, q)$ be the geodesic distance between $p$ and $q$ defined by $$\tilde{d}(p, q)=\arccos(p\cdot q)\leq \pi.$$
For a fixed point $q$ on the surface of $S^2$, the spherical circle is defined as 
$$\tilde{c}_q=\{p\in S^2: \tilde{d}(p, q)\leq r_i\}$$
which is the circle where the center is at the point $q$ with radius $r_i$ and $0\leq r_i < \pi/2$.

\subsection{Problem Formulations}
Assume that the set of radii $\mathcal{R}=\{r_1, r_2, ..., r_n\}$ is given. We place a point $v_i$ on the plane in a way that the distance between $O$ and $v_i$ is the radius $r_i$. Therefore, a simple polygon $P$ is formed from the counterclockwise sequence of vertices $\{v_1, ..., v_n\}$ generated by the sequence of radii $\mathcal{R}$.

In this study, we are interested in the following question.
For a given set of radii $\mathcal{R}=\{r_1, r_2, ..., r_n\}$, does there exist a convex configuration of vertices set $V=\{v_1, ..., v_n\}$ including $O$ with respect to the set of radii $\mathcal{R}$? To avoid the confusion with the problems in \cite{Aich2001,Lenhart1995}, the convex configuration in this context means that the radius $r_i$ is fixed for all $i$, and length of edge $l_i:=d(v_{i}, v_{i+1})$ are allowed to be adjusted with respect to the position of $v_i$ and $r_i$.

In the three-dimensional case, the concept of convex configuration in our context can be considered similar to the two-dimensional case. We assume that a vertex $v_i$ is in $\mathbb{R}^3$ with the Euclidean distance between $O$ and $v_i$, say $r_i$. The convex configuration of the three-dimensional case is defined by the existence of a convex polyhedron whose all of the points in the set $\mathcal{V}$ are vertices of the convex polyhedron. Therefore, the problem in three-dimensional space is to consider the existence of a convex configuration of $v_1, ..., v_n$ with respect to the given radii set $\mathcal{R}$.

\section{Existence of a Convex Polygon in the Plane}

For a set $V=\{v_1, v_2, ..., v_n\}$ of $n>3$ vertices in the plane, we would like to investigate the convexity of the constructed polygon. 

For a given sequence of radii $\mathcal{R}$, if the radii are distinct, the convex configuration can always exist by the following theorem.

\begin{lemma}\label{ThmDistinct}
	Let $\mathcal{R}=\{r_1, ..., r_n\}$ be a given radii set such that $r_i > 0$ and $r_i \neq r_j$ for all $i, j$. Assume that $V=\{v_1, ..., v_n\}$ is a set of vertices induced by $\mathcal{R}$.
	There exists a convex configuration of $\mathcal{V}$ with respect to the radii set $\mathcal{R}$.
\end{lemma}
\begin{proof}
	Without loss of generality, we order the set $\mathcal{R}$ as the descending order, i.e. $r'_1 > r'_2 > ... > r'_n$. Therefore, the set $\{r'_1, r'_2, ..., r'_n\}$ is the strictly decreasing sequence.
	
	We construct a sequence of concentric circles $\mathcal{C}=\{C_1, ..., C_n\}$ such that $C_i=C(O,r'_i)$ is a circle with radius $r'_i$ where the center is $O$.
	
	Since $r_i$ are distinct positive numbers, there exists a line $\ell$ passing through all concentric circles $C_1, C_2, ..., C_{n-1}$, but does not pass through the circle $C_{n}$. Let $\ell^{\perp}$ be a perpendicular line of $\ell$ at $O$. Remark that the circle $C_n$ is laid in a half-plane of $H(\ell)$
	
	With the line $\ell^{\perp}$, choose an arbitrary half-plane $H(\ell^{\perp})$. The vertices $v_1, v_2$ $, .., v_{n-1}$ are chosen by the the intersection of the circle $C_i$ for all $i=1, ..., n-1$, and the line $\ell$ which are laid inside the half-plane $H(\ell^{\perp})$. 
	
	%เลือก v_n บนวงกลม C_n ที่ทำให้ constructed polygon include O
	% คือ เชื่อมจุด v_1, v_n-1, v_n ที่ทำให้เป็นรูปสี่เหลี่ยมที่รวมจุด O ไปด้วย
	
	Let $M$ be a midpoint of the segment $v_1$ and $v_{n-1}$. Draw a line $MO$. Then the last vertex $v_n$ is chosen at the intersection of $MO$ and the circle $C_n$ which is in the other half-plane $H(\ell^{\perp})$, as shown in Figure \ref{fig02-concircs}. Hence, $O$ is in the triangle $\triangle v_1v_{n-1}v_n$ which implies that $O$ is laid inside the polygon constructed in the processes of vertices $\{v_1,v_2,...,v_n\}$. This concludes the proof of the existence of the convex configuration.
	
	\begin{figure}[ht]
		\begin{center}
			\includegraphics[scale=1]{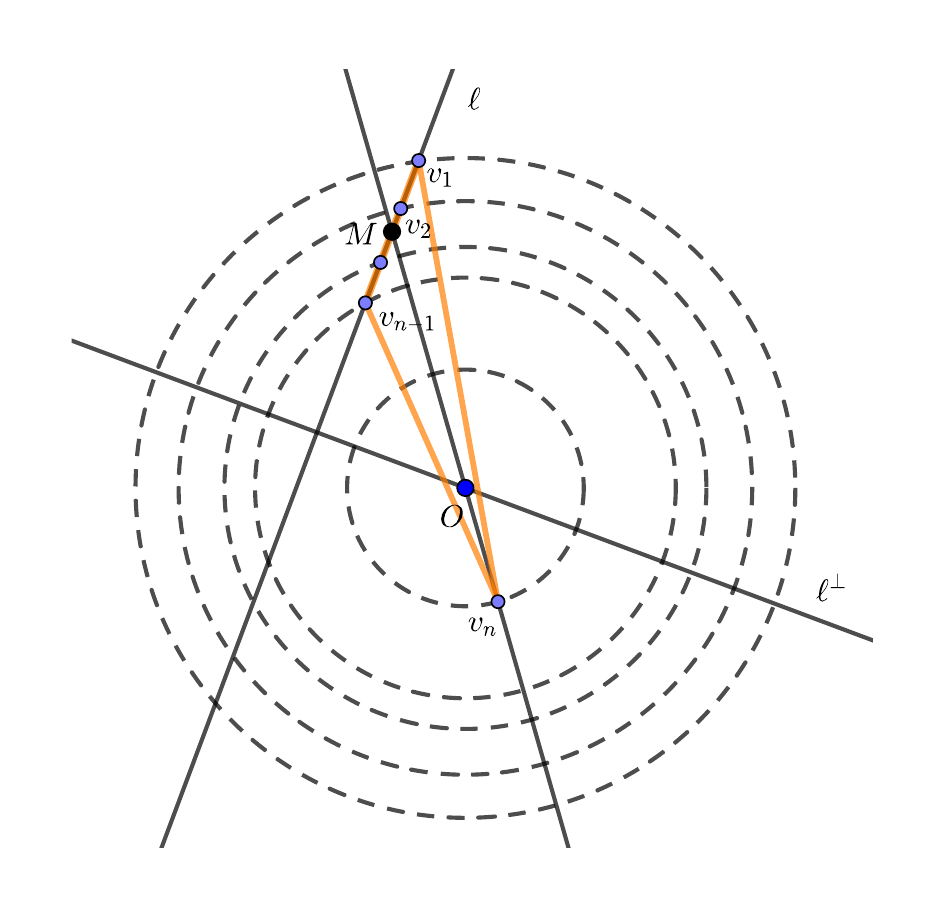}
		\end{center}
		\caption{The construction of a convex polygon with respect to the given distinct radii}\label{fig02-concircs}
	\end{figure}
	
\end{proof}

Remark that in Theorem \ref{ThmDistinct}, the vertices $v_1, ..., v_n$ are allowed to be collinear. In the case of strictly convex configuration, we can perturb the vertices to be non-collinear. Therefore, the following theorem is obtained.

\begin{theorem}\label{LemStrictlyConv}
	For a given distinct positive radii set $\mathcal{R}=\{r_1, ..., r_n\}$ with induced vertices set $V=\{v_1, ..., v_n\}$,
	There exists a strictly convex configuration of $V$ with respect to the radii set $\mathcal{R}$.
\end{theorem}
\begin{proof}
	Assume that the vertices of a convex configuration are located by the processes in Theorem \ref{ThmDistinct} as shown in Figure \ref{fig02-concircs}. The perturbation is done with the vertices $v_3, ..., v_{n-1}$ by the following processes.
	
	We firstly consider the angular distance between vertices $v_1$ and $v_{n-1}$. For the triangle $\triangle v_1v_{n-1}v_n$, the angle $\gamma$ between $\overrightarrow{v_nv_1}$ and $\overrightarrow{v_nv_{n-1}}$ is $\gamma:= \arccos\left(\dfrac{r_1^2 + r^2_{n-1}-d(v_1,v_{n-1})^2}{2r_1r_{n-1}}\right)$, and the angle $\zeta$ between $\overrightarrow{v_nM}$ and $\overrightarrow{v_nv_{n-1}}$ is $\zeta:=\arccos\left(\dfrac{d(v_n,M)^2 + r^2_{n-1}-d(M,v_{n-1})^2}{2d(v_n, M)r_{n-1}}\right)$. Remark that for the vertex $v_{n-1}$, it should not be moved in the region of the region of $v_1v_nM$ to make a polygon $P$ containing the origin $O$. Therefore, the angular movement of all vertices on its circle should be smaller than $\theta:=\gamma-\zeta$.
	
	For the pair of vertices $v_1, v_2$, we draw a ray $\overrightarrow{v_1v_2}$. Therefore, the vertex $v_3$ should be perturbed on the left-handed side of the ray $\overrightarrow{v_1v_2}$ on the circle $C_{r'_3}$ for a circular distance $\epsilon_1$ with angle $\theta>\epsilon_1/r'_3>0$ and move all vertices $v_4, ..., v_{n-1}$ along the ray $\overrightarrow{v_2v_3}$, says $v'_4, ..., v'_{n-1}$. Next, we fix a ray $\overrightarrow{v_2v_3}$ and perturb the vertex $v'_4$ to the left side of the ray $\overrightarrow{v_2v_3}$ for a circular distance $\epsilon_2$ on the circle $C_{r'_4}$ with angle $\theta - \epsilon_1/r'_3 >\epsilon_2/r'_4>0$, says $v''_4$, and move other points $v'_5, ..., v'_{n-1}$ on along the ray $\overrightarrow{v_3v_4}$. 
	
	We continue these processes until all of vertices $v_3, ..., v_{n-1}$ are perturbed such that $$\dfrac{\epsilon_1}{r'_3}+\dfrac{\epsilon_2}{r'_4}+..+\dfrac{\epsilon_{n-3}}{r'_{n-1}} < \theta.$$
	Hence, the resulting polygon is perturbed to be a strictly convex configuration, which concludes the proof of the existence.  
\end{proof}

Before we prove the following lemma, we would define the segment from the intersection between a line and all concentric circles. Let $\ell$ be a line, and $C_1$ be a circle with radius $r_1$ which is the largest circle among the concentric circles. $\overline{\ell}$ is the segment induced from the intersection between $\ell$ and $C_1$, whose the initial and end points are on the circle $C_1$.

With the similar strategy in Lemma \ref{LemStrictlyConv}, we can extend to the case that some radii are same, and the repeated number of the radii is at most 4

\begin{lemma}\label{LemmaRepeated4}
Let $\mathcal{R}=\{r_{(1,1)},..., r_{(1,m_1)},r_{(2,1)},...,r_{(2,m_2)}, ..., r_{(k,1)},..., r_{(k,m_k)}\}$ be a set of radii such that $r_{(i,1)}=...= r_{(i,m_i)}$ for each $i=1, ..., k$ and $1\leq m_i \leq 4$. Then there exist a convex configuration $V$ with respect to the radii set $\mathcal{R}$.
\end{lemma}
\begin{proof}
	Let $V$ be a set of vertices such that each vertex $v_{(i, j)}$ satisfying the radius $r_{(i, j)}$.
	We assume that the elements in $\mathcal{R}$ are sorted in such a way that
	$$r_{(1,1)}=...=r_{(1,m_1)}>r_{(2,1)}=...=r_{(2,m_2)}> ...> r_{(k,1)}=...= r_{(k,m_k)}.$$
	We already proved the case $m_i=1$ for all $i$ in Lemma \ref{ThmDistinct}. Similar to Theorem \ref{ThmDistinct}, the proof relies on the location of points on the concentric circles with radii $r_{(1,1)}, r_{(2,1)}, ..., r_{(k,1)}$. Hence, without loss of generality, assume that the center of circles are at the origin $O$ of $XY$-plane.
	
	We separate the proof into two cases as follows.
		\begin{figure}[ht]
			\begin{center}
				\includegraphics[scale=0.55]{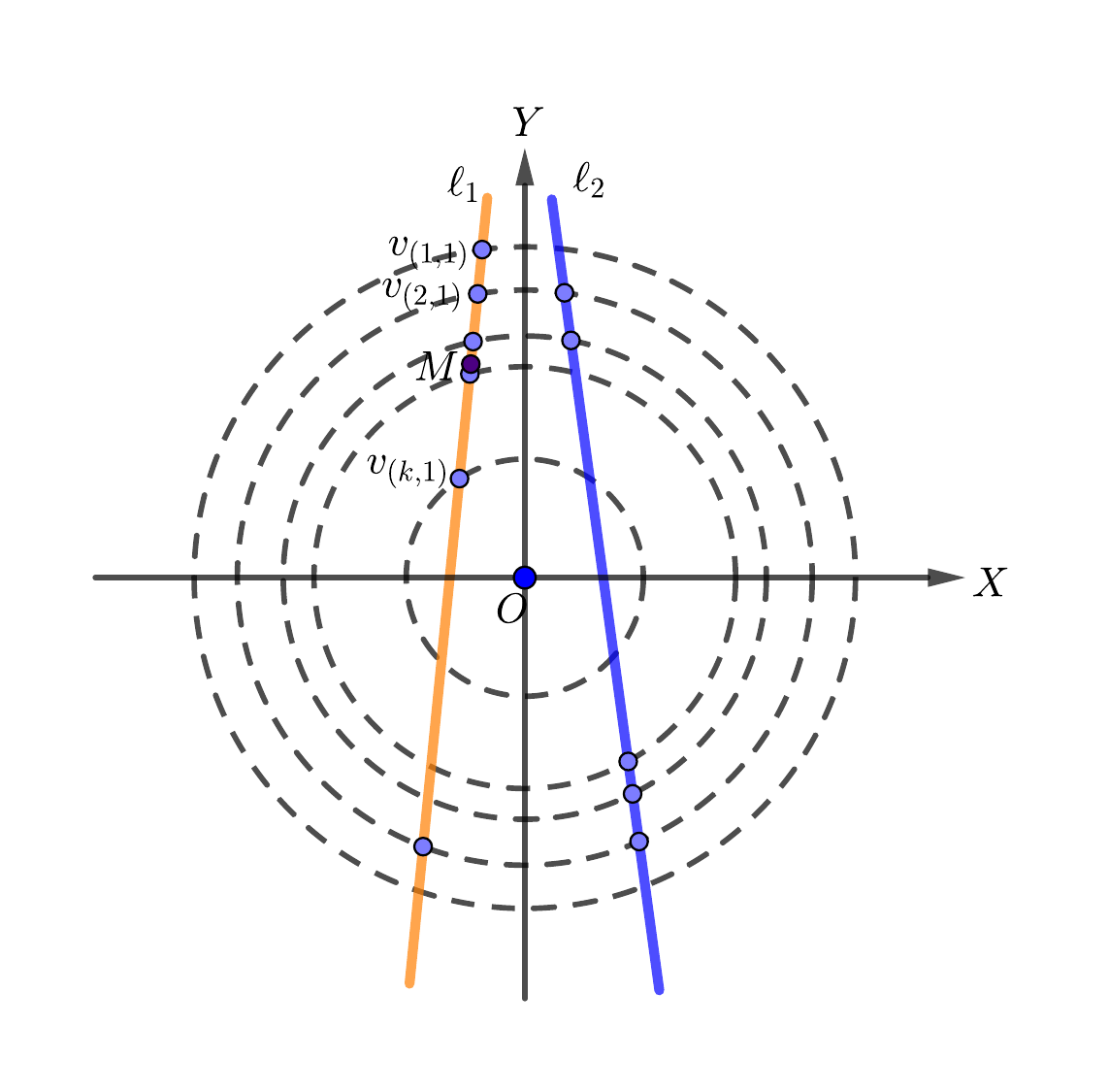} 				\includegraphics[scale=0.55]{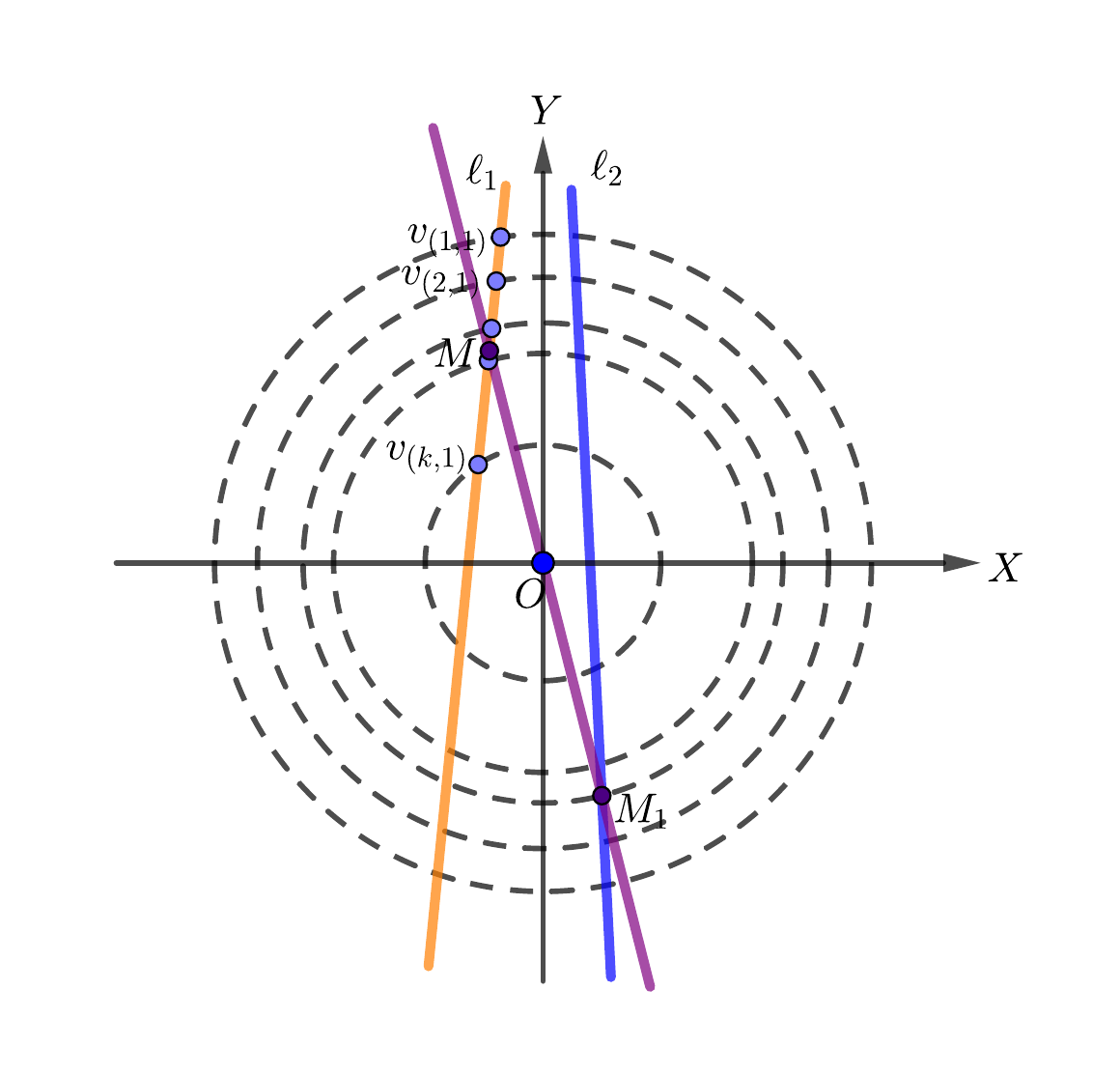}
			\end{center}
			\caption{The construction for the convex configuration of $V$ when (left) $\max\{m_1, ..., m_k\}=4$ and (right) $\max\{m_1, ..., m_k\}=2$ or $3$}\label{fig02-concircsAB}
		\end{figure}
		
	\textbf{Case 1} $\max\{m_1, ..., m_k\}=4$

	We construct a line $\ell_1$ and $\ell_2$ which intersect all concentric circles such that $\overline{\ell_1}$ and $\overline{\ell_2}$ are on the opposite half-plane with respect to the $Y$-axis as shown in Figure \ref{fig02-concircsAB} (left).
	
	Then we lay the points in the set $V$ satisfying each circle radius on the intersection between $\ell_1, \ell_2$, and the concentric circles. This forms a convex quadrilateral containing the origin, which is a convex configuration of $V$ with respect to the given radii.

	\textbf{Case 2} $\max\{m_1, ..., m_k\}=2$ or $3$
	
	Assume that $\max\{m_1, ..., m_k\}=m_p$ such that $p < j$ for all $j=p+1, ..., n$.
	
	We firstly draw a line $\ell_1$ such that $\overline{\ell_1}$ is on a side of a half-plane with respect to $Y$-axis. The first $k$ points $v_{(1, 1)}, v_{(2, 1)}, ..., v_{(k, 1)}$ are chosen from the intersection between $\ell_1$ and concentric circles in a same quadrant. After that, we find the midpoint $M$ between $v_{(1, 1)}$ and $v_{(k, 1)}$ on the line $\ell_1$ and draw a line $\overline{MO}$. The intersection between $\overline{MO}$ and the circle $C_p$ with the radius $r_{(p,1)}$ is denoted as $M_1$. Then, we draw a line $\ell_2$ passing through $M_1$, where $\overline{\ell_2}$ is laid in the opposite half-plane of $\overline{\ell_1}$ with respect to $Y$-axis, and $\ell_2$ intersects all concentric circles, which is shown in Figure \ref{fig02-concircsAB} (right).
	
	Hence, we place the remaining points on the intersections between $\ell_2$ and concentric circles. Since $M_1$ is the point on the largest circle whose contains the maximum number of points, at least $v_{(1, 1)}v_{(k, 1)}M_1$ forms a triangle, or a convex quadrilateral $v_{(1, 1)}v_{(k, 1)}M_1v_{(q_1, q_2)}$ for some $q_1, q_2$, where $v_{(q_1, q_2)}$ is a point on the line $\ell_2$. This forms a convex configuration of $V$ with respect to given radii $\mathcal{R}$.
	
	With these cases, the proof is concluded as desired.
	
\end{proof}

For the strictly convex configuration, we can employ a similar strategy to Theorem \ref{LemStrictlyConv} as presented in the following theorem.

\begin{theorem}\label{ThmStrictlyConv234}
For a given positive radii set $\mathcal{R}=\{r_{(1,1)},...,$ $r_{(1,m_1)},r_{(2,1)},...,$ $r_{(2,m_2)}, ..., r_{(k,1)},..., r_{(k,m_k)}\}$ such that $r_{(i,1)}=...= r_{(i,m_i)}$ for each $i=1, ..., k$ and $1\leq m_i \leq 4$. Then there exist a strictly convex configuration $V$ with respect to the radii set $\mathcal{R}$.
\end{theorem}
\begin{proof}
Suppose that the convex configuration is settled by Lemma \ref{LemmaRepeated4}. A proof relies on on each case as presented in Lemma \ref{LemmaRepeated4}. 

	\textbf{Case 1} $\max\{m_1, ..., m_k\}=4$
	
In this case, the construction in Lemma \ref{LemmaRepeated4} yields a convex quadrilateral. Without loss of generality, assume that all points are separated into 4 quadrants, and suppose to start from points in the second quadrant.

	We firstly move the  $v_{(k, 1)}$ to the position which is close to the negative side of X-axis, i.e. the angle between $\overrightarrow{Ov_{(k, 1)}}$ and the negative side of X-axis is $\theta_1$. Then we draw the line $\ell_{k, 1}$ passing through $v_{(k, 1)}$ and perpendicular to X-axis. Suppose that $H^0(\ell_{k, 1})$ is the half-plane including the origin. We perturb all vertices $v_{(2, 1)}, v_{(3, 1)}, ..., v_{(k-1, 1)}$ by the technique similar to Theorem \ref{LemStrictlyConv} in such a way that all vertices are in the region $H^0(\ell_{k, 1})\backslash \ell_{(k, 1)}$. Therefore, there exists a line $\ell_{1}$ passing through $v_{(k-1, 1)}$ and $v_{(k, 1)}$ which is different to $\ell_{k, 1}$.
	
	Let $V_3=\{v_{(k_j, 2)}:\text{ for some } k_j = 1, ..., k\}$ be a set of points in the third quadrant. Choose the point $v^m_{(k_3, 2)}\in V_3$ such that $r_{(k_3, 2)}=\min\{r_{(k_j, 2)}:\text{ for some } k_j = 1, ..., k\}$ and $v^M_{(k_3, 2)}\in V_3$ such that $r_{(k_3, 2)}=\max\{r_{(k_j, 2)}:\text{ for some } k_j = 1, ..., k\}$. We firstly move $v^m_{(k_3, 2)}$ to the line $\ell_{k, 1}$ and then move $v^M_{(k_3, 2)}$ to the position which is close  to the negative side of Y-axis, i.w. the angle between $\overrightarrow{Ov^M_{(k_3, 1)}}$ and the negative side of y-axis is $\theta_2$. After that, we construct a line $\ell^M_{k_3, 2}$ passing through $v^M_{(k_3, 2)}$ and perpendicular to Y-axis. Then we perturb all points in $V_3$ except $v^m_{(k_3, 2)}$ and $v^M_{(k_3, 2)}$ using the same technique in Theorem \ref{LemStrictlyConv} such that all vertices are laid in the region $(H^0(\ell_{k, 1})\backslash \ell_{(k, 1)})\cap(H^0(\ell^M_{k_3, 2})\backslash\ell^M_{k_3, 2})\cap Q_3$, where $H^0(\ell^M_{k_3, 2})$ is a half-plane of the line $\ell^M_{k_3, 2}$ including the origin and $Q_3$ is the region of the third quadrant.
	
			\begin{figure}[ht]
				\begin{center}
					\includegraphics[scale=0.9]{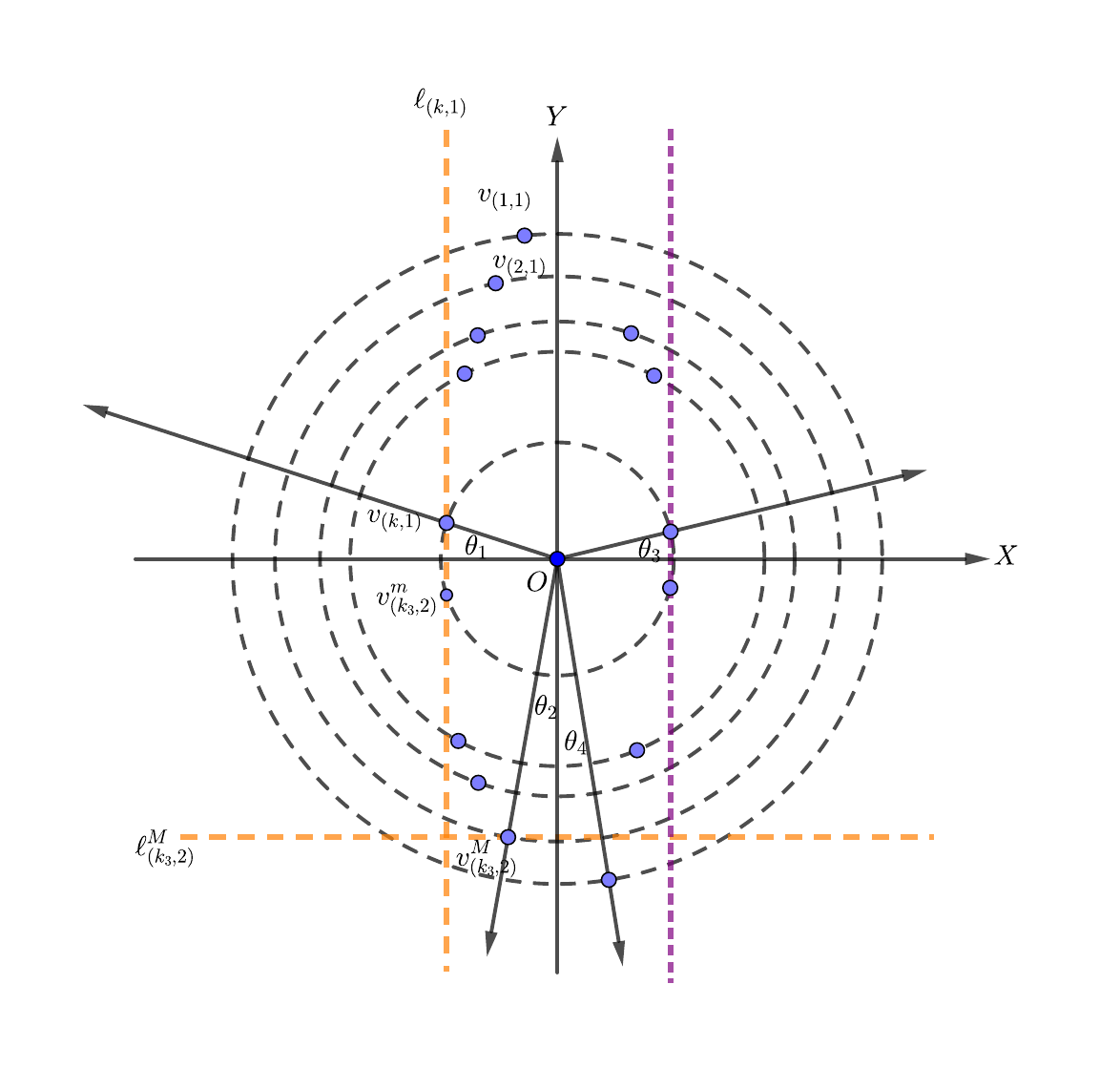} 	
				\end{center}
				\caption{The perturbation of the points to find the strictly convex configuration when $\max\{m_1, ..., m_k\}=4$}\label{fig04-concircsA}
			\end{figure}

	Using a similar technique, we can perturb all vertices in the first quadrant and the fourth quadrant. Finally, the convex polygon can be closed by the point with the largest radius in the third quadrant and fourth quadrant, and the point with the largest radius in the first quadrant and the second quadrant. That is, there is a strictly convex configuration from the given set of radii. 
	
	\textbf{Case 2} $\max\{m_1, ..., m_k\}=2$ or $3$
	
	We can employ the same strategy of the first case to the points in the second, fourth and the first quadrant to obtain the strictly convex configuration of the given set of radii.
	
	Therefore, we can find the strictly convex configuration from the given set of radii in any cases.
\end{proof}

\section{Existence of a Convex Polyhedron in Three-Dimensional Space}

Given a set of radii $\mathcal{R}=\{r_1, ..., r_n\}$, assume that the radii are the Euclidean distance from the origin to the vertex $v_1, ..., v_n$ in three-dimensional space. Recall that the convex configuration, in this case, is the convex polyhedron including the origin $O$.

In the three-dimensional case, the existence of the convex configuration can be proved. Firstly, we consider the simple case when all of the radii are distinct.

\begin{lemma}
	For $n\geq 4$, given a set of positive radii $\mathcal{R}=\{r_1, ..., r_n\}$ such that all of radii are distinct. There exists a convex configuration of $\mathcal{V}=\{v_1, ..., v_n\}$ in three-dimensional space.
\end{lemma}
\begin{proof}
	When $n=4$, we place the point $v_1, v_2, v_3, v_4$ with respect to $r_1, r_2, r_3, r_4$ as vertex of the tetrahedron. Therefore, the convex configuration is obviously obtained.
	
	Suppose that $n\geq 5$.    Assume that the descending order of $\mathcal{R}=\{r_1, ..., r_n\}$ is $\mathcal{R}'=\{r'_1, ..., r'_n\}$, where $r'_i = r_j$ for some $i, j$. Construct concentric spheres $S_1(O, r'_1), S_2(O, r'_2)$ at the origin $O$ with radius $r'_1$ and $r'_2$. Without loss of generality, we place the vertex $v_1$ and $v_2$ on the north pole of sphere $S_1$ and south pole of sphere $S_2$, respectively.
	
	We consider the $XY$-plane and place vertices $v_3, ..., v_n$ onto the $XY$-plane by the processes in Theorem \ref{ThmDistinct} and Lemma \ref{LemStrictlyConv} to obtain a convex polygon $P$ of $\{v_3, ..., v_n\}$. Then we join the edge $v_1$ from the north pole to the vertex set $\{v_3, ..., v_n\}$, and $v_2$ from the south pole to the same set. The obtained polyhedron is a polyhedron whose faces are triangles. Since the polygon $P$ is convex and contain the origin $O$, the constructed polyhedron is convex as desired.\smartqed
	
\end{proof}

In general, the radii set $\mathcal{R}$ is not necessarily distinct. Assume that the set of radii consists of $n$ elements with distinct $k$ elements. Let $\mathcal{R}=\{r_{(1,1)},..., r_{(1,m_1)}$ $,r_{(2,1)},...,r_{(2,m_2)}, ..., r_{(k,1)},..., r_{(k,m_k)}\}$ be a set of radii such that  that
$$r_{(1,1)}=...= r_{(1,m_1)}>r_{(2,1)}=...=r_{(2,m_2)}> ...> r_{(k,1)}=...= r_{(k,m_k)}$$
and $m_1 + m_2 + ... + m_k = n$. 
That is, for the $i$-th layer, the radius of the $i$-th layer is $r_{(i, 1)}$, and the $i$-th layer consists of $m_i$ points.

The following theorem shows the existence of a convex configuration in the three-dimensional case.

\begin{theorem}
	Let $\mathcal{R}$ be a set of radii consisting of $n$ elements with $m_i$ repeated radii for each $i$ distinct radius such that the radii are arranged as
$$r_{(1,1)}=...= r_{(1,m_1)}>r_{(2,1)}=...=r_{(2,m_2)}> ...> r_{(k,1)}=...= r_{(k,m_k)}$$
	and $m_1 + m_2 + ... + m_k = n$. Then there exists a convex configuration of $\mathcal{V}=\{v_1, ..., v_n\}$ induced by the set $\mathcal{R}$.\label{Thm3DCase}
\end{theorem}
\begin{proof}
	We firstly construct $k$ concentric spheres $S_1, S_2, ..., S_k$ where the center is at $O$ with radii $r_{(1,1)}, r_{(2,1)},$ $... , r_{(k,1)}$. 
	
	Let $S_0$ be a sphere whose the radius is $r_{(1, 1)}+\epsilon$ for any $\epsilon>0$. Therefore, there exists a circular cone $\mathcal{C}$ whose  apex $A$ is at the north pole of the sphere $S_0$, and the lateral of the cone intersect all of concentric spheres, i.e. the apex angle $\theta$ satisfies $\theta< 2\arctan\left(\dfrac{r_{(1, 1)}+\epsilon}{r_{(k, 1)}}\right)$.
	
	Hence, the cone $\mathcal{C}$ intersects the concentric spheres $S_1, S_2,..., S_k$ such that the intersection between $\mathcal{C}$ and a sphere $S_j$ for all $j=1, 2, ..., k-1$ is a spherical circle, says $\tilde{c}_{j}$ where the centers are north pole of each sphere. Remark that the distance from $O$ to a point on the circle $\tilde{c}_j$ is $r_{(j, 1)}$. Therefore, there are $k$ circles from the largest sphere $S_1$ to the smallest sphere $ S_k$ on the upper hemisphere, as shown as the cross section in Figure \ref{fig03-conspheres}.

		\begin{figure}[ht]
		\begin{center}
			\includegraphics[scale=1]{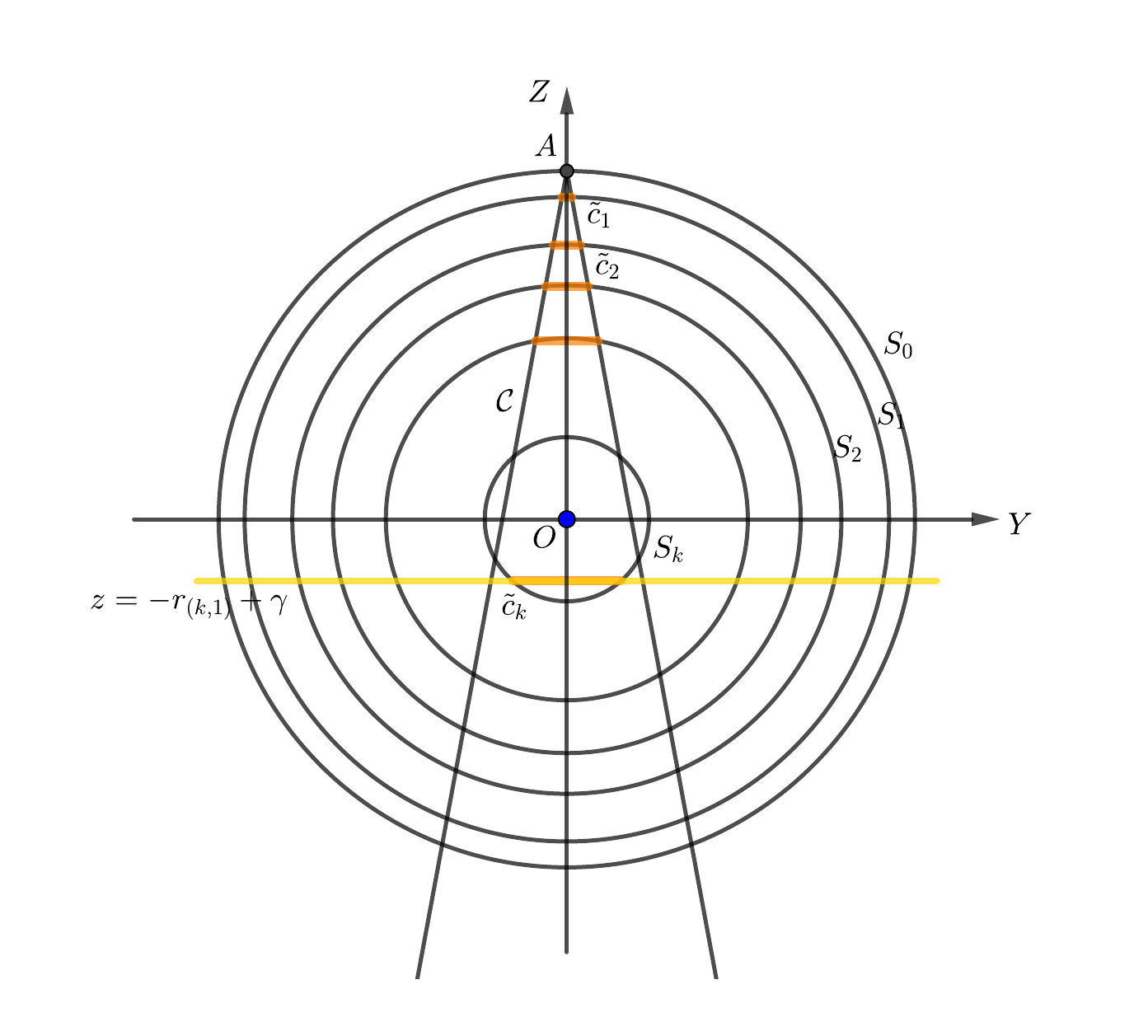}
		\end{center}
		\caption{The cross section at $YZ$-plane for the concentric spheres including a spherical circle of each layer and the cone $\mathcal{C}$}\label{fig03-conspheres}
	\end{figure}

	We choose a line $\ell$ emanating from the apex $A$ on the surface of $\mathcal{C}$. For each layer of circle $\tilde{c}_j$ over the upper hemisphere except the smallest layer $\tilde{c}_k$, place a point on the line $\ell$. Therefore, each layer has at least one point on its layer.
	
	For the number of points of $k$ layers $m_1, m_2, ..., m_k$, assume that the $p$-th layer contains the maximum number of points, i.e. $m_p:=\max\{m_1, m_2, ..., m_k\}$. Therefore, we firstly distribute $m_p$ points on the spherical circle $\tilde{c}_p$ at the $p$-th layer in such a way that the angular distance $\beta$ between each vertex on the spherical circle $\tilde{c}_p$ are equal. Note that we fix the point which is already placed on the line $\ell$ and distribute other $m_p-1$ points, says $v_{p, 1}, v_{p, 2}, ..., v_{p, m_p}$. 
	
	For each $v_{p, i}$ on the $p$-level, construct a plane $P_{v_{p, i}}$ passing through $v_{p, i}$ and $Z$-axis to create a spherical grid. Remark that $P_{v_{p, i}}$ intersects all concentric spheres and generate longitude lines $L_{v_{p, i},\tilde{c}_p}$ on each sphere $S_i$ which are great circles.
	
	Therefore, in each level $j=1, ..., p-1, p+1, ..., k$, the latitude is considered as the spherical circle $\tilde{c}_j$ which intersects longitude $L_{v_{p, j},\tilde{c}_p}$ to $m_p$ points. We can place $m_j$ points on those intersections arbitrarily since $m_j\leq m_p$.
	Since all vertices are laid on the convex surface, for each placed points on the intersections, there exists a plane tangent to the cone passing through that point, and all points are in the same side of the plane. Hence, there exist faces joining $v_{p_1, i_1}, v_{p_2, i_2}, v_{p_3, i_3}$ for some $p_1, i_1, p_2, i_2, p_3, i_3$ which form faces of convex polyhedra.
	
	With the exceptional case for the last smallest layer, say the $k$-th level, we construct a plane $z=-r_{(k, 1)}+\gamma$. Therefore, the parameter $\gamma$ can be considered in the following case.
	
	If the $k$-th layer contains exactly one point, choose $\gamma = 0$. This means that the plane $z=-r_{m_1+...+m_{k-1}+1}$ is a tangent plane at $(0, 0, -r_{(k, 1)})$. Therefore, the polyhedron can be bounded by joining all of the vertices to that point.
	
	Otherwise, assume that there are $m_k$ points at the $k$-th layer. We choose a small $\gamma > 0$ such that $\gamma < |r_{(k, 1)}|$. Therefore, there exists a spherical circle in the $k$-th layer. Then, we distribute $m_k$ points with the same angle and connecting the points in $k$-th level to the above levels to construct a convex polyhedron.
	
	Therefore, the convex configuration exists by the construction process, which concludes the proof.

\end{proof}

\section{Applications}
The main application of the existence of the convex configuration in the three-dimensional case is the confirmation about the non-emptiness properties of the spherical Laguerre Voronoi diagram, which the details will be described soon.

We first recall the definitions and constructions of the spherical Laguerre Voronoi diagram as presented in \cite{Sugihara2002}.

On the unit sphere $S^2$, Let $P=\{p_1, ..., p_n\}$ be a set of points with the weight wet $W=\{w_1, ..., w_n\}$ and $\mathcal{G}=\{\tilde{c}_1, ..., \tilde{c}_n\}$ be a set of spherical circles whose each center is a point in $P$ corresponding to a weight in $W$. The spherical Laguerre Voronoi diagram $\mathcal{L}=\{L_1, ..., L_n\}$ is a Voronoi diagram generated from the set of spherical circles $\mathcal{G}$ with the Laguerre proximity $\tilde{d}_L(c_i, p)=\dfrac{\cos(\tilde{d}(p, p_i))}{\cos(w_i)}$, for a point $p\in S^2$.

The algorithms for constructing the spherical Laguerre Voronoi diagram presented in \cite{Sugihara2002} were based on the intersection of half-spaces of planes passing through the spherical circles including the origin. The dual structure of the spherical Laguerre Voronoi diagram is the spherical Laguerre Delaunay diagram. 

The spherical Laguerre Delaunay diagram can be constructed by the following procedures. For a set of generating circles $\mathcal{G}$, suppose that $P_i$ be a plane passing through the spherical circle $\tilde{c}_i$. Therefore, the dual point of the plane $P_i$ can be considered as $P_i^*=\dfrac{1}{\cos w_i}p_i$, and the spherical Laguerre Delaunay diagram can be constructed from the central projection of the convex hull $\mathcal{G}^*=\{P_1^*, ..., P_n^*\}$ onto the unit sphere $S^2$.

For a spherical Laguerre Voronoi diagram $\mathcal{L}$ generated by $\mathcal{G}$, the spherical Laguerre Voronoi cell $L_i$ is said to be \textit{empty} if $L_i=\emptyset$. The spherical Laguerre Voronoi diagram $\mathcal{L}$ satisfies the \textit{non-emptiness property} if for all $i$, $L_i\neq \emptyset$. Remark that a cell $L_i$ of the spherical Laguerre Voronoi diagram is empty if the dual point $P^*_i$ of the circle $\tilde{c}_i$ is inside of the convex hull of the set $\mathcal{G}^*$.

Instead of giving the spherical circles, assume that the radii of the spherical circles are given. The interesting question is to consider whether or not we can find the location of generators on the sphere in such a way that the generated spherical Laguerre Voronoi diagram satisfies the non-emptiness property.

The answer to the mentioned question is positive as follows.

\begin{theorem}
	Let $W=\{w_1, ..., w_n\}$ be a set of spherical circle radii. Then there exists a spherical Laguerre Voronoi diagram satisfying the non-emptiness property.
\end{theorem}

\begin{proof}
For the set of spherical circle radii $W=\{w_1, ..., w_n\}$, each radius corresponds to the radius $r_i=1/\cos(w_i)$. Remark that $r_i\geq 1$ by the assumption of the spherical circle radius. Therefore, we generate a set of radii $\mathcal{R}=\{r_1, ..., r_n\}$.

By Theorem \ref{Thm3DCase}, there exists a convex configuration of a set $\mathcal{P}=\{p_1, ..., p_n\}$ with respect to $\mathcal{R}$. Therefore, all of dual points in $\mathcal{G}^*$ are on the corner of the convex hull of $\mathcal{G}^*$. That is, the spherical Laguerre Delaunay diagram with respect to $\mathcal{G}^*$ consists all of generators $\mathcal{P}=\{p_1, ..., p_n\}$. 

Hence, it implies there exists a spherical Laguerre Voronoi diagram satisfying non-emptiness property with respect to the given set of radii as desired.
\end{proof}

\section{Concluding Remarks}
We consider the convex configuration problem of $n$ points when the radii which are measured from the fixed point are given. In the two-dimensional case, we have proved that the strictly convex configuration always exists when all radii are distinct or each radius is repeated at most four points. However, the problem is still open when repeated radii are greater than or equal to five points. Therefore, we leave a conjecture to prove this interesting property.

\textbf{Conjecture:} \textit{For any set of given radii $\mathcal{R}$, it is not always to find the convex configuration with respect to the given set $\mathcal{R}$.}

However, the existence of a convex configuration is guaranteed in the case of the three-dimensional space. Using this fact, we can apply the existence of a convex configuration to the existence of the spherical Laguerre Voronoi diagram satisfying the non-emptiness property.

\begin{acknowledgements}
We would like to thank Masaki Moriguchi and Vorapong Suppakitpaisan for some discussions. We also thank the Japan Student Services Organization (JASSO) for the FYI2018 follow-up research fellowship to support the stay of the first author in Japan during this study. This research was supported by Chiang Mai University, Thailand.

\end{acknowledgements}

% BibTeX users please use one of
%\bibliographystyle{spbasic}      % basic style, author-year citations
%\bibliographystyle{spmpsci}      % mathematics and physical sciences
%\bibliographystyle{spphys}       % APS-like style for physics
%\bibliography{}   % name your BibTeX data base

% Non-BibTeX users please use

\end{document}